\documentclass[10pt,letterpaper]{article}
\usepackage[top=0.85in,left=2.75in,footskip=0.75in]{geometry}

\usepackage{amsmath,amssymb}

\usepackage{changepage}

\usepackage{charter}
\usepackage{inconsolata}
\usepackage[T1]{fontenc}
\usepackage[utf8]{inputenc}

\usepackage{textcomp,marvosym}

\usepackage{cite}

\usepackage{nameref}
\usepackage[unicode=true]{hyperref}

\usepackage[right]{lineno}

\usepackage{microtype}
\DisableLigatures[f]{encoding = *, family = * }

\usepackage[table]{xcolor}

\usepackage{array}

\newcolumntype{+}{!{\vrule width 2pt}}

\newlength\savedwidth%

\newcommand\thickhline{\noalign{\global\savedwidth\arrayrulewidth\global\arrayrulewidth 1.5pt}%
\hline%
\noalign{\global\arrayrulewidth\savedwidth}}

\usepackage{setspace} 

\raggedright%
\usepackage[
  aboveskip=1pt,
  labelfont=bf,
  labelsep=period,
  justification=raggedright,
  singlelinecheck=off
]{caption}

\makeatletter
\renewcommand{\@biblabel}[1]{\quad#1.}
\makeatother

\date{}

\usepackage{lastpage,fancyhdr,graphicx}
\usepackage{epstopdf}
\fancyheadoffset[L]{2.25in}
\fancyfootoffset[L]{2.25in}

\usepackage[numbib,nottoc]{tocbibind}

\usepackage[capitalise]{cleveref}

\usepackage[inline]{enumitem}

\usepackage{graphicx,grffile}
\makeatletter
\def\maxwidth{\ifdim\Gin@nat@width>\linewidth\linewidth\else\Gin@nat@width\fi}
\def\maxheight{\ifdim\Gin@nat@height>\textheight\textheight\else\Gin@nat@height\fi}
\makeatother
\setkeys{Gin}{width=\maxwidth,height=\maxheight,keepaspectratio}

\usepackage{mathtools}

\PassOptionsToPackage{hyphens}{url} 
\hypersetup{%
  pdftitle={PyPhi: A toolbox for integrated information theory},
  colorlinks=true,
  linkcolor=blue,
  citecolor=blue,
  urlcolor=blue,
  breaklinks=true
}
\usepackage[super]{nth}

\usepackage{amsthm}
\makeatletter
\def\th@plain{%
  \thm@notefont{}
  \itshape 
}
\def\th@definition{%
  \thm@notefont{}
  \normalfont
}
\makeatother

\newtheorem*{theorem*}{Theorem}

\newtheorem*{definition*}{Definition}

\usepackage{xspace}

\usepackage{fancyvrb}

\DefineVerbatimEnvironment{Highlighting}{Verbatim}{commandchars=\\\{\},fontsize=\small}

\newcommand*{\ie}{\emph{i.e.}\xspace}
\newcommand*{\eg}{\emph{e.g.}\xspace}
\newcommand*{\iit}{\textsc{IIT}\xspace}
\newcommand*{\ces}{\textsc{CES}\xspace}
\newcommand*{\tpm}{\textsc{TPM}\xspace}
\newcommand*{\cm}{\textsc{CM}\xspace}
\newcommand*{\mip}{\textsc{MIP}\xspace}

\newcommand*{\mics}{\textsc{MICS}\xspace}
\newcommand*{\mic}{\textsc{MIC}\xspace}
\newcommand*{\mie}{\textsc{MIE}\xspace}
\newcommand*{\otherwise}{\text{otherwise}}
\newcommand*{\by}{\hspace{0.05em}{\times}\hspace{0.05em}}
\newcommand*{\cut}{\text{cut}}
\newcommand*{\emd}{\textsc{EMD}\xspace}
\DeclareMathOperator*{\emdmath}{\operatorname{\textsc{EMD}}}
\DeclareMathOperator*{\cmmath}{\operatorname{\textsc{CM}}}
\DeclareMathOperator*{\er}{\operatorname{\textsc{ER}}}

\newcommand*{\effectrepertoire}{{\texttt{effect\_repertoire}}\xspace}
\newcommand*{\causerepertoire}{{\texttt{cause\_repertoire}}\xspace}

\newcommand*{\computesia}{{\texttt{pyphi.compute.sia()}}\xspace}
\newcommand*{\computemajorcomplex}{{\texttt{pyphi.compute.major\_complex()}}\xspace}
\newcommand*{\pyphiconcept}{{\texttt{Concept}}\xspace}
\newcommand*{\pyphiconcepts}{{\texttt{Concepts}}\xspace}
\newcommand*{\pyphices}{{\texttt{CauseEffectStructure}}\xspace}
\newcommand*{\pyphisia}{{\texttt{SystemIrreducibilityAnalysis}}\xspace}
\newcommand*{\pyphiria}{{\texttt{RepertoireIrreducibilityAnalysis}}\xspace}

\newcommand*{\sectionref}[1]{\S\,\nameref{#1}}
\newcommand*{\siref}[1]{\nameref{#1}}

\begin{document}

\begin{flushleft}
{\Large
\textbf{PyPhi: A toolbox for integrated information theory}
}
\newline
\\
William G.\ P.\ Mayner\textsuperscript{1,2,*},
William Marshall\textsuperscript{2},
Larissa Albantakis\textsuperscript{2},
Graham Findlay\textsuperscript{1,2},
Robert Marchman\textsuperscript{2},
Giulio Tononi\textsuperscript{2,*}
\\
\bigskip
\footnotesize{
\begin{itemize}[leftmargin=*,itemsep=0pt,topsep=0pt,labelsep=1pt]
\item[\textsuperscript{1}] Neuroscience Training Program, University of Wisconsin--Madison, Madison, WI, USA
\item[\textsuperscript{2}] Department of Psychiatry, Wisconsin Institute for Sleep and Consciousness, University of Wisconsin--Madison, Madison, WI, USA
\item[*] \texttt{mayner@wisc.edu}, *\texttt{gtononi@wisc.edu}
\end{itemize}
}
\bigskip
\end{flushleft}

\subsubsection*{Abstract}
Integrated information theory provides a mathematical framework to fully characterize the cause-effect structure of a physical system. Here, we introduce \emph{PyPhi}, a Python software package that implements this framework for causal analysis and unfolds the full cause-effect structure of discrete dynamical systems of binary elements. The software allows users to easily study these structures, serves as an up-to-date reference implementation of the formalisms of integrated information theory, and has been applied in research on complexity, emergence, and certain biological questions. We first provide an overview of the main algorithm and demonstrate PyPhi's functionality in the course of analyzing an example system, and then describe details of the algorithm's design and implementation.

PyPhi can be installed with Python’s package manager via the command `\texttt{pip~install~pyphi}' on Linux and macOS systems equipped with Python 3.4 or higher. PyPhi is open-source and licensed under the GPLv3; the source code is hosted on GitHub at \texttt{\small\url{https://github.com/wmayner/pyphi}}. Comprehensive and continually-updated documentation is available at \texttt{\small\url{https://pyphi.readthedocs.io}}. The \texttt{pyphi-users} mailing list can be joined at \texttt{\small\url{https://groups.google.com/forum/\#!forum/pyphi-users}}. A web-based graphical interface to the software is available at \texttt{\small\url{http://integratedinformationtheory.org/calculate.html}}.

\hypertarget{sec:introduction}{%
\section{Introduction}\label{sec:introduction}
}

Integrated information theory (\iit) has been proposed as a theory of consciousness. The central hypothesis is that a physical system has to meet five requirements (`postulates') in order to be a physical substrate of subjective experience:
\begin{enumerate*}[label={(\arabic*)}]
  \item \emph{intrinsic existence} (the system must be able to make a difference to itself); 
  \item \emph{composition} (it must be composed of parts that have causal power within the whole);
  \item \emph{information} (its causal power must be specific);
  \item \emph{integration} (its causal power must not be reducible to that of its parts); and
  \item \emph{exclusion} (it must be maximally irreducible)
\end{enumerate*}~\cite{tononi2016integrated,tononi2015integrated,oizumi2014phenomenology,balduzzi2008integrated,tononi2004information}.

From these postulates, \iit develops a mathematical framework to assess the cause-effect structure (\ces) of a physical system that is applicable to discrete dynamical systems. This framework has proven useful not only for the study of consciousness but has also been applied in research on complexity~\cite{albantakis2015intrinsic,albantakis2014evolution,oizumi2016unified,albantakis2017automata}, emergence~\cite{hoel2016can,hoel2013quantifying,marshall2016black}, and certain biological questions~\cite{marshall2017how}.

The main measure of cause-effect power, \emph{integrated information} (denoted \(\Phi\)), quantifies how irreducible a system's \ces is to those of its parts. \(\Phi\) also serves as a general measure of complexity that captures to what extent a system is both integrated~\cite{albantakis2015intrinsic} and differentiated (informative)~\cite{marshall2016integrated}. 

Here we describe \emph{PyPhi}, a Python software package that implements {\iit}'s framework for causal analysis and unfolds the full \ces of discrete dynamical systems of binary elements. The software allows users to easily study these {\ces}s and serves as an up-to-date reference implementation of the formalisms of \iit.

Details of the mathematical framework are published elsewhere~\cite{oizumi2014phenomenology,tononi2016integrated}; in \sectionref{sec:results} we describe the output and input of the software and give an overview of the main algorithm in the course of reproducing results obtained from an example system studied in~\cite{oizumi2014phenomenology}. In \sectionref{sec:design} we discuss specific issues concerning the algorithm's implementation. Finally in \sectionref{sec:availability} we describe how the software can be obtained and discuss new functionality planned for future versions.

\hypertarget{sec:results}{%
\section{Results}\label{sec:results}
}

\hypertarget{sec:output}{%
\subsection{Output}\label{sec:output}
}

The software has two primary functions:
\begin{enumerate*}[label={(\arabic*)}]
  \item to unfold the full \ces of a discrete dynamical system of interacting elements and compute its \(\Phi\) value, and
  \item to compute the maximally-irreducible cause-effect repertoires of a particular set of elements within the system.
\end{enumerate*}
The first is function is implemented by \computemajorcomplex, which returns a \pyphisia object (\cref{fig:object-hierarchy}A). The system's \ces is contained in the `\texttt{ces}' attribute and its \(\Phi\) value is contained in `\texttt{phi}'. Other attributes are detailed in the \href{https://pyphi.readthedocs.io/page/api/models.subsystem.html\#pyphi.models.subsystem.SystemIrreducibilityAnalysis}{online documentation}.

The \ces is composed of \pyphiconcept objects, which are the output of the second main function: \texttt{Subsystem.concept()} (\cref{fig:object-hierarchy}B). Each \pyphiconcept is specified by a set of elements within the system (contained in its `\texttt{mechanism}' attribute). A \pyphiconcept contains a maximally-irreducible cause and effect repertoire (`\causerepertoire' and `\effectrepertoire'), which are probability distributions that capture how the mechanism elements in their current state constrain the previous and next state of the system, respectively; a \(\varphi\) value (`\texttt{phi}'), which measures the irreducibility of the repertoires; and several other attributes discussed below and detailed in the \href{https://pyphi.readthedocs.io/page/api/models.mechanism.html\#pyphi.models.mechanism.Concept}{online documentation}.

\hypertarget{fig:object-hierarchy}{%
\begin{figure}[!h]
\centering
\begin{adjustwidth}{-2in}{0in}
\includegraphics[width=7.25in]{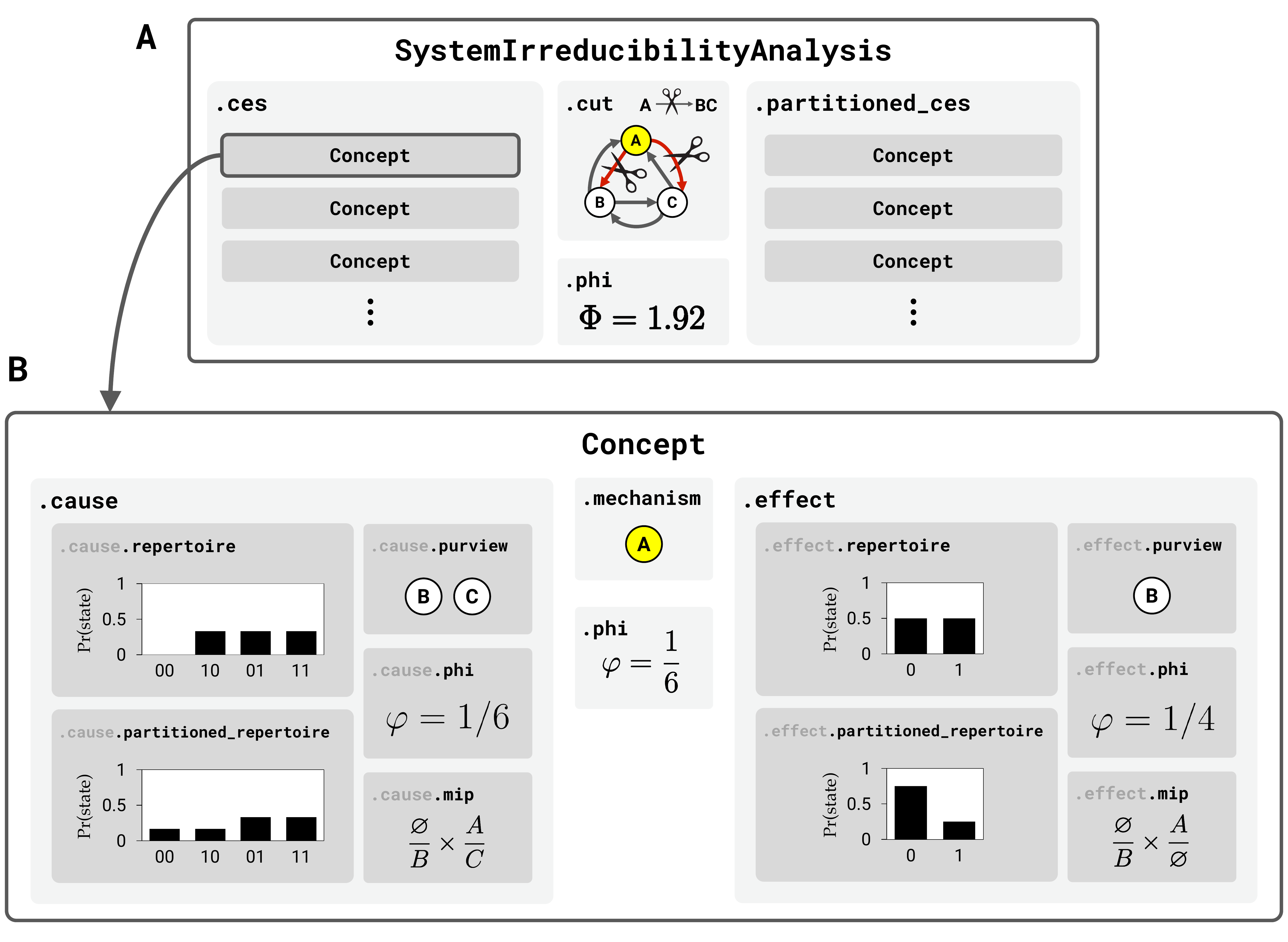}
\end{adjustwidth}
\caption{{\bf Output.} \textbf{(A)}\ The \pyphisia object is the main output of the software. It represents the results of the analysis of the system in question. It has several attributes (grey boxes): `\texttt{ces}' is a \pyphices object containing all of the system's \pyphiconcepts; `\texttt{cut}' is a \texttt{Cut} object that represents the minimum-information partition (\mip) of the system (the partition of the system that makes the least difference to its \ces); `\texttt{partitioned\_ces}' is the \pyphices of \pyphiconcepts specified by the system after applying the \mip; and `\texttt{phi}' is the \(\Phi\) value, which measures the difference between the unpartitioned and partitioned \ces. \textbf{(B)}\ A \pyphiconcept represents the maximally-irreducible cause (\mic) and maximally-irreducible effect (\mie) of a mechanism in a state. The `\texttt{mechanism}' attribute contains the indices of the mechanism elements. The `\texttt{cause}' and `\texttt{effect}' attributes contain \texttt{MaximallyIrreducibleCause} and \texttt{MaximallyIrreducibleEffect} objects that describe the mechanism's \mic and \mie, respectively; each of these contains a purview, repertoire, \mip, partitioned repertoire, and \(\varphi\) value. The `\texttt{phi}' attribute contains the concept's \(\varphi\) value, which is the minimum of the \(\varphi\) values of the \mic and \mie.}\label{fig:object-hierarchy}
\end{figure}
}

\hypertarget{sec:input}{%
\subsection{Input}\label{sec:input}
}

The starting point for the \iit analysis is a discrete dynamical system \(S\) composed of \(n\) interacting elements. Such a system can be represented by a directed graph of interconnected nodes, each equipped with a Markovian function that outputs the node's state at the next timestep \(t+1\) given the state of its parents at the previous timestep \(t\) (\cref{fig:tpm}). At present, PyPhi can analyze both deterministic and stochastic systems consisting of elements with two states.

\hypertarget{fig:tpm}{%
\begin{figure}[!h]
\centering
\begin{adjustwidth}{-2in}{0in}
\includegraphics[width=7.25in]{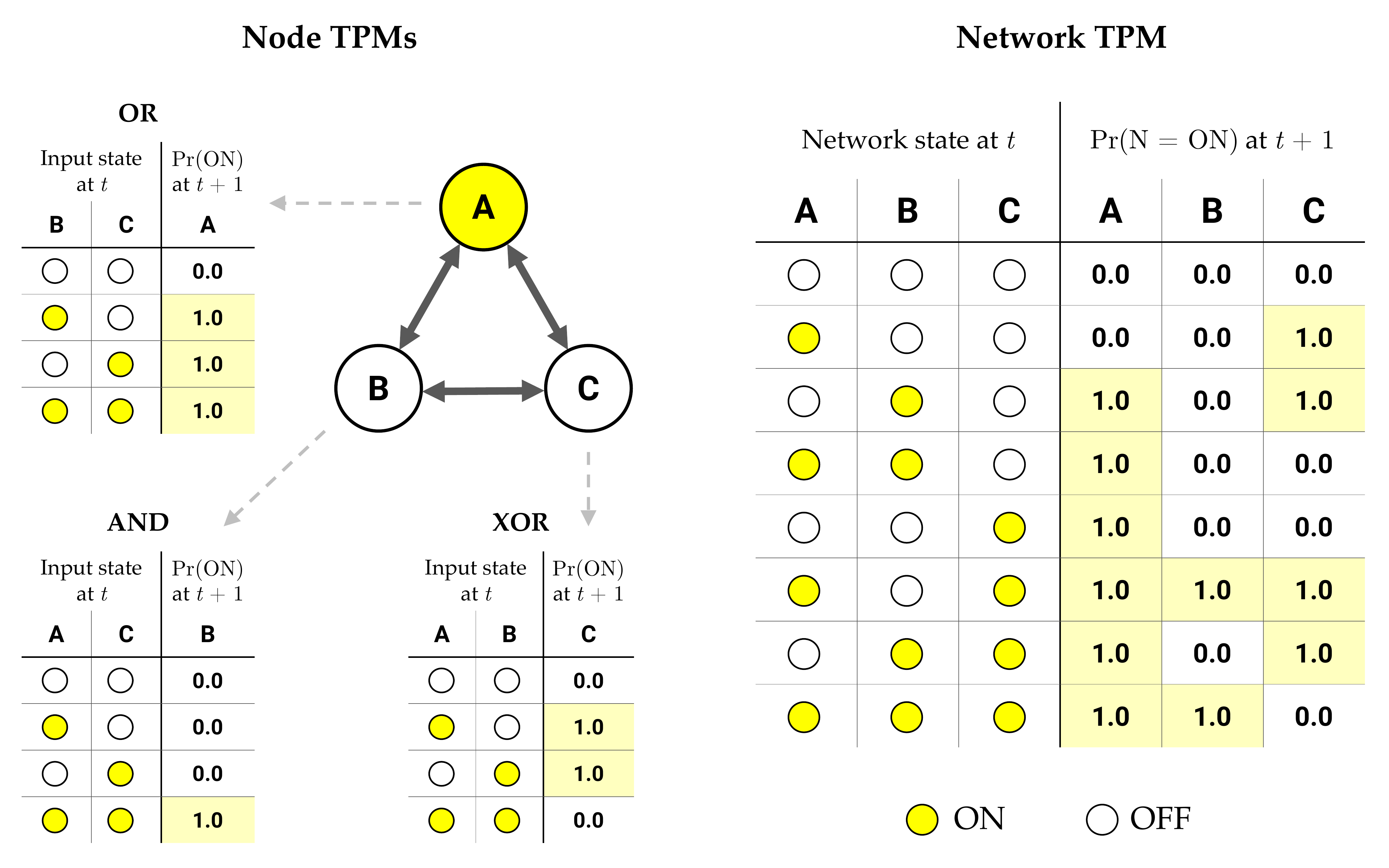}
\end{adjustwidth}
\caption{%
\textbf{A network of nodes and its \tpm.} Each node has its own \tpm---in this case, the truth-table of a deterministic logic gate. Yellow signifies the ``ON'' state; white signifies ``OFF''. The system's \tpm (right) is composed of the {\tpm}s of its nodes (left), here shown in state-by-node form (see \sectionref{sec:tpm}). Note that in PyPhi's \tpm representation, the first node's state varies the fastest, according to the little-endian convention (see \sectionref{sec:2d-sbn-form}).}\label{fig:tpm}
\end{figure}
}

Such a discrete dynamical system is completely specified by its transition probability matrix (\tpm), which contains the probabilities of all state transitions from \(t\) to \(t+1\). It can be obtained from the graphical representation of the system by perturbing the system into each of its possible states and observing the following state at the next timestep (for stochastic systems, repeated trials of perturbation/observation will yield the probabilities of each state transition). In PyPhi, the \tpm is the fundamental representation of the system.

Formally, if we let \(S_t\) be the random variable of the system state at \(t\), the \tpm specifies the conditional probability distribution over the next state \(S_{t+1}\) given each current state \(s_t\): 
\[ 
  \Pr(S_{t+1} \mid S_t = s_t), \quad \forall \; s_t \in \Omega_{S}, 
\] 
where \(\Omega_S\) denotes the set of possible states. Furthermore, given a marginal distribution over the previous states of the system, the \tpm fully specifies the joint distribution over state transitions. Here \iit imposes uniformity on the marginal distribution of the previous states because the aim of the analysis is to capture direct causal relationships across a single timestep without confounding factors, such as influences from system states before \(t - 1\)~\cite{albantakis2017actual,ay2008information,balduzzi2008integrated,hoel2013quantifying,oizumi2014phenomenology}. The marginal distribution thus corresponds to an interventional (causal), not observed, state distribution.

Moreover, \iit assumes that there is no instantaneous causation; that is, it is assumed that the elements of the dynamical system influence one another only from one timestep to the next. Therefore we require that the system satisfies the following Markov condition, called the \emph{conditional independence property}: each element's state at \(t+1\) must be independent of the state of the others, given the state of the system at \(t\)~\cite{pearl2009causality},
\begin{equation}
  \Pr(S_{t+1} \mid S_t = s_t) \;= \prod_{N \,\in\, S} \Pr(N_{t+1} \mid S_t = s_t)
    \;, \quad \forall \; s_t \in S.\label{eq:conditional-independence}
\end{equation}
For systems of binary elements, a \tpm that satisfies \cref{eq:conditional-independence} can be represented in state-by-node form (\cref{fig:tpm}, right), since we need only store each element's marginal distribution rather than the full joint distribution.

In PyPhi, the system under analysis is represented by a \texttt{Network} object. A \texttt{Network} is created by passing its \tpm as the first argument: \texttt{network\ =\ pyphi.Network(tpm)} (see \sectionref{sec:demo-setup}). Optionally, a connectivity matrix (\cm) can also be provided, where 
\[
  {\left[\cmmath\right]}_{i,j} =
    \begin{cases}
      1 &\text{if there is an edge from element }i\text{ to element }j \\
      0 &\otherwise,
    \end{cases}
\]
via the \texttt{cm} keyword argument: \texttt{network\ =\ pyphi.Network(tpm,\ cm=cm)}. Because the \tpm completely specifies the system, providing a \cm is not necessary; however, explicit connectivity information can be used to make computations more efficient, especially for sparse networks, because PyPhi can rule out certain causal influences \emph{a priori} if there are missing connections (see \sectionref{sec:connectivity}). Note that this means providing an incorrect \cm can result in inaccurate output. If no \cm is given, PyPhi assumes full connectivity; \ie, it assumes each element may have an effect on any other, which guarantees correct results.

Once the \texttt{Network} is created, a subset of elements within the system (called a \emph{candidate system}), together with a particular system state, can be selected for analysis by creating a \texttt{Subsystem} object. Hereafter we refer to a candidate system as a \emph{subsystem}.

\hypertarget{sec:demonstration}{%
\subsection{Demonstration}\label{sec:demonstration}
}

The mathematical framework of \iit is typically formulated using graphical causal models as representations of physical systems of elements. The framework builds on the causal calculus of the \(\textit{do}(\,\cdot\,)\) operator introduced by Pearl~\cite{pearl2009causality}. In order to assess causal relationships among the elements, interventions (manipulations, perturbations) are used to actively set elements into a specific state, after which the resulting state transition is observed.

For reference, we define a set of graphical operations that are used during the \iit analysis. To \emph{fix} an element is to use system interventions to keep it in the same state for every observation. To \emph{noise} an element is to use system interventions to set it into a state chosen uniformly at random. Finally, to \emph{cut} a connection from a source element to a target element is to make the source appear noised to the target, while the remaining, uncut connections from the source still correctly transmit its state.

In this section we demonstrate some of the capabilities of the software by unfolding the \ces of a small deterministic system of logic gates as described in~\cite{oizumi2014phenomenology} while describing how the algorithm is implemented in terms of \tpm manipulations, which we link to the graphical operations defined above. A schematic of the algorithm is shown in \cref{fig:algorithm-mechanism} and \cref{fig:algorithm-system}, and a more detailed illustration is given in \siref{sup:presentation}.

\hypertarget{fig:algorithm-mechanism}{%
\begin{figure}[!h]
\centering
\begin{adjustwidth}{-2in}{0in}
\includegraphics[width=7.25in]{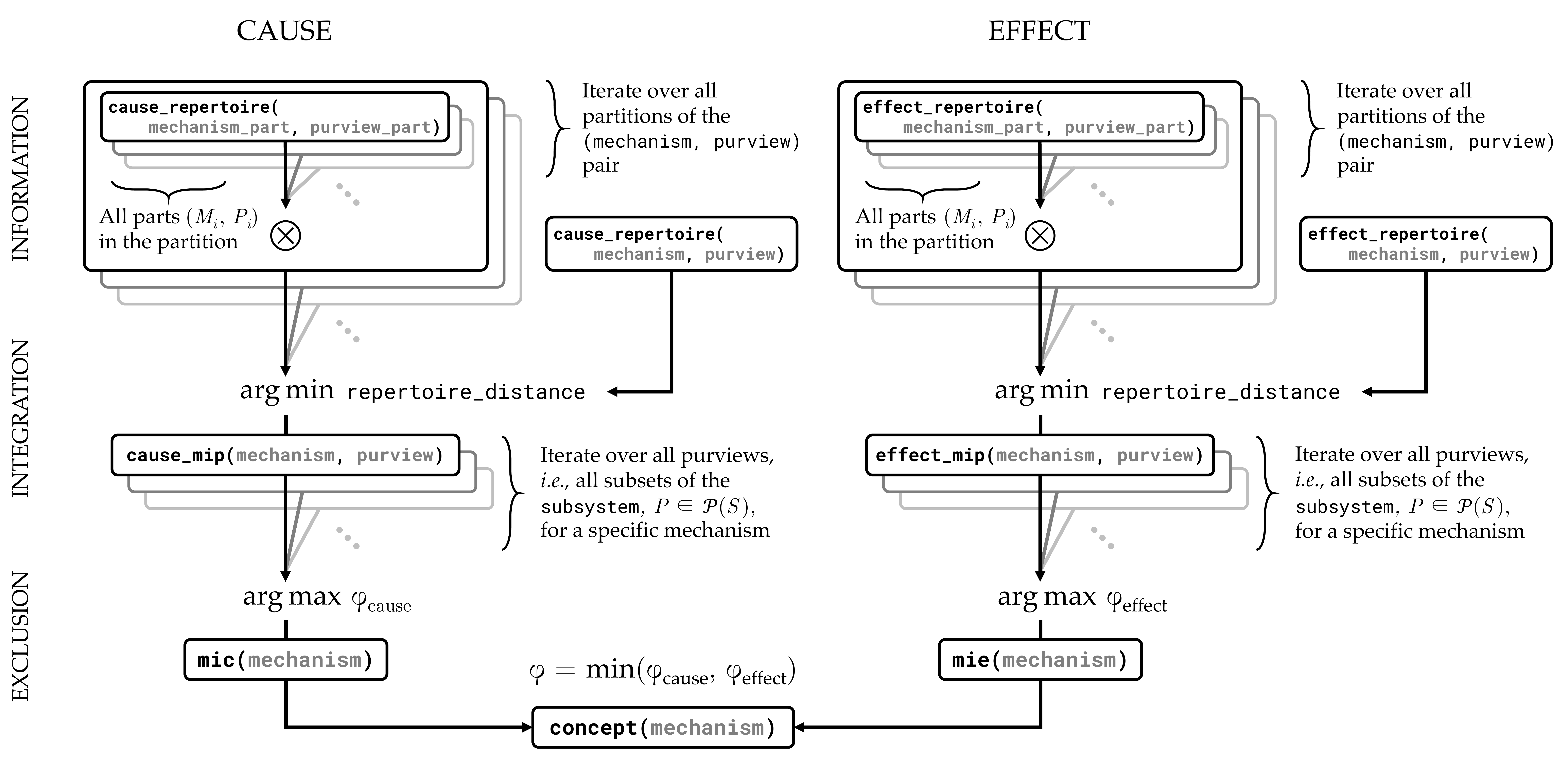}
\end{adjustwidth}
\caption{\textbf{Algorithm schematic at the mechanism level.} PyPhi functions are named in boxes, with arguments in grey. Arrows point from callee to caller. Functions are organized by the postulate they correspond to (left). \(\otimes\) denotes the tensor product; \(\mathcal{P}\) denotes the power set.}\label{fig:algorithm-mechanism}
\end{figure} 
}
\hypertarget{fig:algorithm-system}{%
\begin{figure}[!h]
\centering
\includegraphics[width=5.5in]{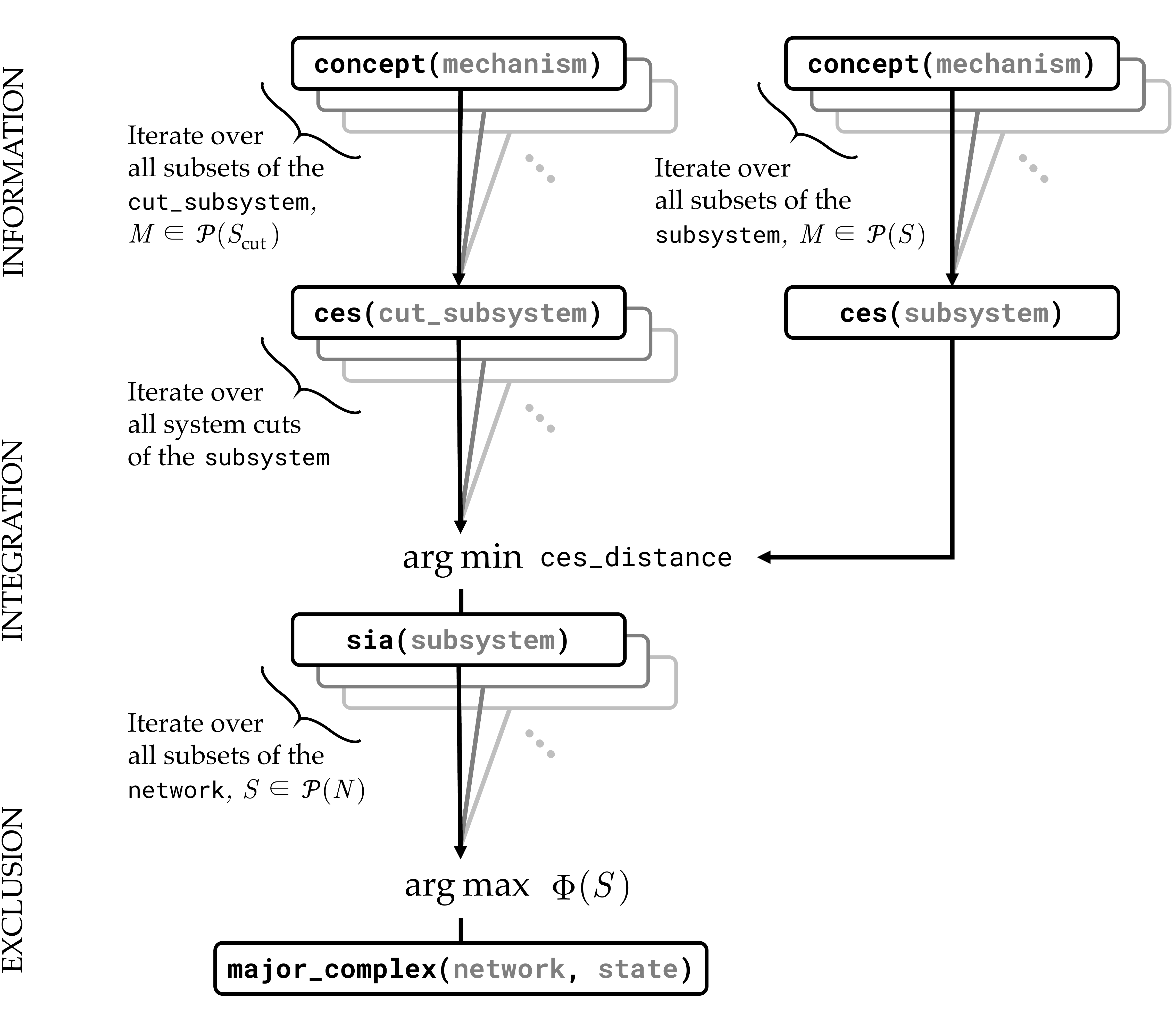}
\caption{\textbf{Algorithm schematic at the system level.} PyPhi functions are named in boxes, with arguments in grey. Arrows point from callee to caller. Functions are organized by the postulate they correspond to (left). \(\mathcal{P}\) denotes the power set.}\label{fig:algorithm-system}
\end{figure}
}

\hypertarget{sec:demo-setup}{%
\subsubsection{Setup}\label{sec:demo-setup}
}

The first step is to create the \texttt{Network} object. Here we choose to provide the \tpm in \(2\)-dimensional state-by-node form (see \sectionref{sec:2d-sbn-form}). The \tpm is the only required argument, but we provide the \cm as well, since we know that there are no self-loops in the system and PyPhi will use this information to speed up the computation. We also label the nodes \(A\), \(B\), and \(C\) to make the output easier to read.

\begin{samepage}
\begin{Highlighting}
  >>> import pyphi
  >>> import numpy as np
  >>> tpm = np.array([
  ...     [0.0, 0.0, 0.0],
  ...     [0.0, 0.0, 1.0],
  ...     [1.0, 0.0, 1.0],
  ...     [1.0, 0.0, 0.0],
  ...     [1.0, 0.0, 0.0],
  ...     [1.0, 1.0, 1.0],
  ...     [1.0, 0.0, 1.0],
  ...     [1.0, 1.0, 0.0]
  ... ])
  >>> cm = np.array([
  ...     [0, 1, 1],
  ...     [1, 0, 1],
  ...     [1, 1, 0],
  ... ])
  >>> network = pyphi.Network(tpm, cm=cm, node_labels=[{\textquotesingle}A{\textquotesingle}, {\textquotesingle}B{\textquotesingle}, {\textquotesingle}C{\textquotesingle}])
\end{Highlighting}
\end{samepage}

We select a subsystem and a system state for analysis by creating a \texttt{Subsystem} object. System states are represented by tuples of \texttt{1}s and \texttt{0}s, with \texttt{1} meaning ``ON'' and \texttt{0} meaning ``OFF.'' In this case we will analyze the entire system, so the subsystem will contain all three nodes. The nodes to include can be specified with either their labels or their indices (note that in other PyPhi functions, nodes must be specified with their indices).

\begin{samepage}
\begin{Highlighting}
  >>> state = (1, 0, 0)
  >>> nodes = ({\textquotesingle}A{\textquotesingle}, {\textquotesingle}B{\textquotesingle}, {\textquotesingle}C{\textquotesingle})
  >>> subsystem = pyphi.Subsystem(network, state, nodes)
\end{Highlighting}
\end{samepage}

If there are nodes outside the subsystem, they are considered as \emph{background conditions} for the causal analysis~\cite{oizumi2014phenomenology}. In the graphical representation of the system, the background conditions are \emph{fixed} in their current state while the subsystem is perturbed and observed in order to derive its \tpm. In the \tpm representation, the equivalent operation is performed by \emph{conditioning} the system \tpm on the state at \(t\) of the nodes outside the subsystem and then \emph{marginalizing out} those nodes at \(t+1\) (illustrated in \siref{sup:presentation}). In PyPhi, this is done when the subsystem is created; the subsystem \tpm can be accessed with the \texttt{tpm} attribute, \eg \texttt{subsystem.tpm}.

\hypertarget{sec:demo-repertoire}{%
\subsubsection{Cause/effect repertoires (mechanism-level information)}\label{sec:demo-repertoire}
}

The lowest-level objects in the \ces of a system are the \emph{cause repertoire} and \emph{effect repertoire} of a set of nodes within the subsystem, called a \emph{mechanism}, over another set of nodes within the subsystem, called a \emph{purview} of the mechanism. The cause (effect) repertoire is a probability distribution that captures the information specified by the mechanism about the purview by describing how the previous (next) state of the purview is constrained by the current state of the mechanism.

In terms of graphical operations, the effect repertoire is obtained by 
\begin{enumerate*}[label={(\arabic*)}]
  \item \emph{fixing} the mechanism nodes in their state at \(t\); 
  \item \emph{noising} the non-mechanism nodes at time \(t\), so as to remove their causal influence on the purview; and 
  \item observing the resulting state transition from \(t\) to \(t+1\) while ignoring the state at \(t+1\) of non-purview nodes, in order to derive a distribution over purview states at \(t+1\).
\end{enumerate*}

The cause repertoire is obtained similarly, but in that case, the purview nodes at time \(t-1\) are noised, and the resulting state transition from \(t-1\) to \(t\) is observed while ignoring the state of non-mechanism nodes. Bayes' rule is then applied, resulting in a distribution over purview states at \(t-1\). The corresponding operations on the \tpm are detailed in \sectionref{sec:repertoires} and illustrated in \siref{sup:presentation}.

Note that, operationally, we enforce that each input from a noised node conveys \emph{independent} noise during the perturbation/observation step. In this way, we avoid counting correlations from outside the mechanism-purview pair as constraints due to the current state of the mechanism. Graphically, this process would correspond to replacing each noised node that is a parent of multiple purview nodes (for the effect repertoire) or mechanism nodes (for the cause repertoire) with multiple, independent ``virtual nodes'' (as in~\cite[Supplementary Methods]{oizumi2014phenomenology}). However, the equivalent definition of repertoires in \cref{eq:multinode-effect-repertoire,eq:multinode-cause-repertoire} obviates the need to actually implement virtual nodes in PyPhi.

With the \texttt{cause\_repertoire()} method of the \texttt{Subsystem}, we can obtain the cause repertoire of, for example, mechanism \(A\) over the purview \(ABC\) depicted in Fig.~4 of~\cite{oizumi2014phenomenology}:

\begin{samepage}
\begin{Highlighting}
  >>> A, B, C = subsystem.node_indices
  >>> print(subsystem.state)
  (1, 0, 0)
  >>> mechanism = (A,)
  >>> purview = (A, B, C)
  >>> cr = subsystem.cause_repertoire(mechanism, purview)
  >>> print(cr)
  [[[ 0.          0.16666667]
    [ 0.16666667  0.16666667]]
  
   [[ 0.          0.16666667]
    [ 0.16666667  0.16666667]]]
\end{Highlighting}
\end{samepage}

We see that mechanism \(A\) in its current state, ON (\texttt{1}), specifies information by ruling out the previous states in which \(B\) and \(C\) are OFF (\texttt{0}). That is, the probability that either \texttt{(0, 0, 0)} or \texttt{(1, 0, 0)} was the previous state, given that \(A\) is currently ON, is zero:

\begin{samepage}
\begin{Highlighting}
  >>> print(cr[(0, 0, 0)])
  0.0
  >>> print(cr[(1, 0, 0)])
  0.0
\end{Highlighting}
\end{samepage}

Note that repertoires are returned in multidimensional form, so they can be indexed with state tuples as above. Repertoires can be reshaped to be 1-dimensional if needed, \eg for plotting, but care must be taken that NumPy's FORTRAN (column-major) ordering is used so that PyPhi's little-endian convention for indexing states is respected (see \sectionref{sec:2d-sbn-form}). PyPhi provides the \texttt{pyphi.distribution.flatten()} function for this:

\begin{samepage}
\begin{Highlighting}
  >>> flat_cr = pyphi.distribution.flatten(cr)
  >>> print(flat_cr)
  [ 0.          0.          0.16666667  0.16666667  0.16666667  0.16666667
    0.16666667  0.16666667]
\end{Highlighting}
\end{samepage}

\hypertarget{sec:demo-mip}{%
\subsubsection{Minimum-information partitions (mechanism-level integration)}\label{sec:demo-mip}
}

Having assessed the information of a mechanism over a purview, the next step is to assess its \emph{integrated information} (denoted \(\varphi\)) by quantifying the extent to which the cause and effect repertoires of the mechanism-purview pair can be reduced to the repertoires of its parts.

In terms of graphical operations, the irreducibility of a mechanism-purview pair is tested by partitioning it into parts and \emph{cutting} the connections between them. By applying the perturbation/observation procedure after cutting the connections we obtain a \emph{partitioned repertoire}. Since the partition renders the parts independent, in terms of \tpm manipulations, the partitioned repertoire can be calculated as the product of the individual repertoires for each of the parts. If the partitioned repertoire is no different than the original unpartitioned repertoire, then the mechanism as a whole did not specify integrated information about the purview. By contrast, if a repertoire cannot be factored in this way, then some of its selectivity is due to the causal influence of the mechanism \emph{as an integrated whole} on the purview, and the repertoire is said to be \emph{irreducible.}

The amount of irreducibility of a mechanism over a purview with respect to a partition is quantified as the distance between the unpartitioned repertoire and the partitioned repertoire (calculated with \texttt{pyphi.distance.repertoire\_distance()}). There are many ways to divide the mechanism and purview into parts, so the irreducibility is measured for every partition and the partition that yields the minimum irreducibility is called the \emph{minimum-information partition} (\mip). The integrated information (\(\varphi\)) of a mechanism-purview pair is the distance between the unpartitioned repertoire and the partitioned repertoire associated with the \mip. PyPhi supports several distance measures and partitioning schemes (see \sectionref{sec:config}).

The \mip search procedure is implemented by the \texttt{Subsystem.cause\_mip()} and \texttt{Subsystem.effect\_mip()} methods. Each returns a \pyphiria object that contains the \mip, as well as the \(\varphi\) value, mechanism, purview, temporal direction (cause or effect), unpartitioned repertoire, and partitioned repertoire. For example, we compute the effect \mip of mechanism \(ABC\) over purview \(ABC\) from Fig.~6 of~\cite{oizumi2014phenomenology} as follows:

\begin{samepage}
\begin{Highlighting}
  >>> mechanism = (A, B, C)
  >>> purview = (A, B, C)
  >>> mip = subsystem.effect_mip(mechanism, purview)
  >>> print(mip)
  Repertoire irreducibility analysis
    \(\varphi\) = 1/4
    Mechanism: [A, B, C]
    Purview = [A, B, C]
    Direction: EFFECT
    Partition:
       \(\varnothing\)    A,B,C
      --- × -----
       B     A,C 
    Repertoire:
      +--------------+
      |  S     P(S)  |
      | ------------ |
      | 000    0     |
      | 100    0     |
      | 010    0     |
      | 110    0     |
      | 001    1     |
      | 101    0     |
      | 011    0     |
      | 111    0     |
      +--------------+
    Partitioned repertoire:
      +--------------+
      |  S     P(S)  |
      | ------------ |
      | 000    0     |
      | 100    0     |
      | 010    0     |
      | 110    0     |
      | 001    3/4   |
      | 101    0     |
      | 011    1/4   |
      | 111    0     |
      +--------------+
\end{Highlighting}
\end{samepage}

Here we can see that the \mip attempts to factor the repertoire of \(ABC\) over \(ABC\) into the product of the repertoire of \(ABC\) over \(AC\) and the repertoire of the empty mechanism \(\varnothing\) over \(B\). However, the repertoire cannot be factored in this way without information loss; the distance between the unpartitioned and partitioned repertoire is nonzero (\(\varphi = \frac{1}{4}\)). Thus mechanism \(ABC\) over the purview \(ABC\) is irreducible.

\hypertarget{sec:demo-mice}{%
\subsubsection{Maximally-irreducible cause-effect repertoires (mechanism-level exclusion)}\label{sec:demo-mice}
}

Next, we apply {\iit}'s postulate of exclusion at the mechanism level by finding the \emph{maximally-irreducible cause} (\mic) and \emph{maximally irreducible effect} (\mie) specified by a mechanism. This is done by searching over all possible purviews for the \pyphiria object with the maximal \(\varphi\) value. The \texttt{Subsystem.mic()} and \texttt{Subsystem.mie()} methods implement this search procedure; they return a \texttt{MaximallyIrreducibleCause} and a \texttt{MaximallyIrreducibleEffect} object, respectively. The \mic of mechanism \(BC\), for example, is the purview \(AB\) (Fig.~8 of~\cite{oizumi2014phenomenology}). This is computed like so:

\begin{samepage}
\begin{Highlighting}
  >>> mechanism = (B, C)
  >>> mic = subsystem.mic(mechanism)
  >>> print(mic)
  Maximally-irreducible cause
    \(\varphi\) = 1/3
    Mechanism: [B, C]
    Purview = [A, B]
    Direction: CAUSE
    MIP:
       B     C 
      --- × ---
       \(\varnothing\)    A,B
    Repertoire:
      +-------------+
      | S     P(S)  |
      | ----------- |
      | 00    2/3   |
      | 10    0     |
      | 01    0     |
      | 11    1/3   |
      +-------------+
    Partitioned repertoire:
      +-------------+
      | S     P(S)  |
      | ----------- |
      | 00    1/2   |
      | 10    0     |
      | 01    0     |
      | 11    1/2   |
      +-------------+
\end{Highlighting}
\end{samepage}

\hypertarget{sec:demo-concepts}{%
\subsubsection{Concepts}\label{sec:demo-concepts}
}

If the mechanism's \mic has \(\varphi_{\text{cause}} > 0\) and its \mie has \(\varphi_{\text{effect}} > 0\), then the mechanism is said to specify a \emph{concept}. The \(\varphi\) value of the concept as a whole is the minimum of \(\varphi_{\text{cause}}\) and \(\varphi_{\text{effect}}\).

We can compute the concept depicted in Fig.~9 of~\cite{oizumi2014phenomenology} using the \texttt{Subsystem.concept()} method, which takes a \texttt{mechanism} and returns a \pyphiconcept object containing the \(\varphi\) value, the \mic (in the `\texttt{cause}' attribute), and the \mie (in the `\texttt{effect}' attribute):

\begin{samepage}
\begin{Highlighting}
  >>> mechanism = (A,)
  >>> concept = subsystem.concept(mechanism)
  >>> print(concept)
  ~~~~~~~~~~~~~~~~~~~~~~~~~~~~~~~~~~~~~~~~~~~~~~~~~~~~~~~~
             Concept: Mechanism = [A], \(\varphi\) = 1/6            
  ~~~~~~~~~~~~~~~~~~~~~~~~~~~~~~~~~~~~~~~~~~~~~~~~~~~~~~~~
              MIC                         MIE             
  +--------------------------++--------------------------+
  |  \(\varphi\) = 1/6                 ||  \(\varphi\) = 1/4                |
  |  Purview = [B, C]        ||  Purview = [B]           |
  |  MIP:                    ||  MIP:                    |
  |     \(\varnothing\)     A              ||    \(\varnothing\)     A              |
  |    --- × ---             ||    --- × ---             |
  |     B     C              ||     B     \(\varnothing\)              |
  |  Repertoire:             ||  Repertoire:             |
  |    +-------------+       ||    +------------+        |
  |    | S     P(S)  |       ||    | S    P(S)  |        |
  |    | ----------- |       ||    | ---------- |        |
  |    | 00    0     |       ||    | 0    1/2   |        |
  |    | 10    1/3   |       ||    | 1    1/2   |        |
  |    | 01    1/3   |       ||    +------------+        |
  |    | 11    1/3   |       ||  Partitioned repertoire: |
  |    +-------------+       ||    +------------+        |
  |  Partitioned repertoire: ||    | S    P(S)  |        |
  |    +-------------+       ||    | ---------- |        |
  |    | S     P(S)  |       ||    | 0    3/4   |        |
  |    | ----------- |       ||    | 1    1/4   |        |
  |    | 00    1/6   |       ||    +------------+        |
  |    | 10    1/6   |       |+--------------------------+
  |    | 01    1/3   |       |
  |    | 11    1/3   |       |
  |    +-------------+       |
  +--------------------------+
\end{Highlighting}
\end{samepage}

Note that in PyPhi, the repertoires are distributions over purview states, rather than system states. Occasionally it is more convenient to represent repertoires as distributions over the entire system. This can be done with the \texttt{expand\_cause\_repertoire()} and \texttt{expand\_effect\_repertoire()} methods of the \pyphiconcept object, which assume the unconstrained (maximum-entropy) distribution over the states of non-purview nodes:

\begin{samepage}
\begin{Highlighting}
  >>> full_cr = concept.expand_cause_repertoire()
  >>> print(pyphi.distribution.flatten(full_cr))
  [ 0.          0.          0.16666667  0.16666667  0.16666667  0.16666667
    0.16666667  0.16666667]
  >>> full_er = concept.expand_effect_repertoire()
  >>> print(pyphi.distribution.flatten(full_er))
  [ 0.0625  0.1875  0.0625  0.1875  0.0625  0.1875  0.0625  0.1875]
\end{Highlighting}
\end{samepage}

Also note that \texttt{Subsystem.concept()} will return a \pyphiconcept object when \(\varphi = 0\) even though these are not concepts, strictly speaking. For convenience, \texttt{bool(concept)} evaluates to \texttt{True} if \(\varphi > 0\) and \texttt{False} otherwise.

\hypertarget{sec:demo-ces}{%
\subsubsection{Cause-effect structures (system-level information)}\label{sec:demo-ces}
}

The next step is to compute the \ces, the set of all concepts specified by the subsystem. The \ces characterizes all of the causal constraints that are intrinsic to a physical system. This is implemented by the \texttt{pyphi.compute.ces()} function, which simply calls \texttt{Subsystem.concept()} for every mechanism \(M \in \mathcal{P}(S)\), where \(\mathcal{P}(S)\) is the power set of subsystem nodes. It returns a \pyphices object containing those \pyphiconcepts for which \(\varphi > 0\).

We see that every mechanism in \(\mathcal{P}(S)\) except for \(AC\) specifies a concept, as described in Fig.~10 of~\cite{oizumi2014phenomenology}: 

\begin{samepage}
\begin{Highlighting}
  >>> ces = pyphi.compute.ces(subsystem)
  >>> print(ces.labeled_mechanisms)
  ([{\textquotesingle}A{\textquotesingle}], [{\textquotesingle}B{\textquotesingle}], [{\textquotesingle}C{\textquotesingle}], [{\textquotesingle}A{\textquotesingle}, {\textquotesingle}B{\textquotesingle}], [{\textquotesingle}B{\textquotesingle}, {\textquotesingle}C{\textquotesingle}], [{\textquotesingle}A{\textquotesingle}, {\textquotesingle}B{\textquotesingle}, {\textquotesingle}C{\textquotesingle}])
\end{Highlighting}
\end{samepage}

\hypertarget{sec:demo-ices}{%
\subsubsection{Irreducible cause-effect structures (system-level integration)}\label{sec:demo-ices}
}

At this point, the irreducibility of the subsystem's \ces is evaluated by applying the integration postulate at the system level. As with integration at the mechanism level, the idea is to measure the difference made by each partition and then take the minimal value as the irreducibility of the subsystem.

We begin by performing a \emph{system cut}. Graphically, the subsystem is partitioned into two parts and the edges going from one part to the other are \emph{cut}, rendering them causally ineffective. This is implemented as an operation on the \tpm as follows: Let \(E_{\cut}\) denote the set of directed edges in the subsystem that are to be cut, where each edge \(e \in E_{\cut}\) has a source node \(a\) and a target node \(b\). For each edge, we modify the individual \tpm of node \(b\) (\cref{fig:tpm}) by marginalizing over the states of \(a\) at \(t\). The resulting \tpm specifies the function implemented by \(b\) with the causal influence of \(a\) removed. We then combine the modified node {\tpm}s to recover the full \tpm of the partitioned subsystem. Finally, we recalculate the \ces of the subsystem with this modified \tpm (the \emph{partitioned CES}).

The irreducibility of a \ces with respect to a partition is the distance between the unpartitioned and partitioned {\ces}s (calculated with \texttt{pyphi.compute.ces\_distance()}; several distances are supported; see \sectionref{sec:config}). This distance is evaluated for every partition, and the minimum value across all partitions is the subsystem's integrated information \(\Phi\), which measures the extent to which the \ces specified by the subsystem is irreducible to the \ces under the minimal partition.

This procedure is implemented by the \computesia function, which returns a \pyphisia object (\cref{fig:object-hierarchy}). We can verify that the \(\Phi\) value of the example system in~\cite{oizumi2014phenomenology} is \(1.92\) and the minimal partition is that which removes the causal connections from \(AB\) to \(C\):

\begin{samepage}
\begin{Highlighting}
  >>> sia = pyphi.compute.sia(subsystem)
  >>> print(sia.phi)
  1.916665
  >>> print(sia.cut)
  Cut [0, 1] --/ /--> [2]
\end{Highlighting}
\end{samepage}

\hypertarget{sec:demo-complexes}{%
\subsubsection{Complexes (system-level exclusion)}\label{sec:demo-complexes}}

The final step in unfolding the \ces of the system is to apply the postulate of exclusion at the system level. We compute the {\ces} of each subset of the network, considered as a subsystem (that is, \emph{fixing} the external nodes as background conditions), and find the \ces with maximal \(\Phi\), called the \emph{maximally-irreducible cause-effect structure} (\mics) of the system. The subsystem giving rise to it is called the \emph{major complex}; any overlapping subsets with lower \(\Phi\) are excluded. Non-overlapping subsets may be further analyzed to find additional complexes within the system.

In this example, we find that the whole system \(ABC\) is the system's major complex, and all proper subsets are excluded: 

\begin{samepage}
\begin{Highlighting}
  >>> state = (1, 0, 0)
  >>> major_complex = pyphi.compute.major_complex(network, state)
  >>> print(major_complex.subsystem)
  Subsystem(A, B, C)
\end{Highlighting}
\end{samepage}

Note that since \computemajorcomplex is a function of the \texttt{Network}, rather than a particular \texttt{Subsystem}, it is necessary to specify the state in which the system should be analyzed.

\hypertarget{sec:design}{%
\section{Design and implementation}\label{sec:design}
}

PyPhi was designed to be easy to use in interactive, exploratory research settings while nonetheless remaining suitable for use in large-scale simulations or as a component in larger applications. It was also designed to be efficient, given the high computational complexity of the algorithms in \iit. Here we describe some implementation details and optimizations used in the software.

\hypertarget{sec:tpm}{%
\subsection{Representation of the \tpm and probability distributions}\label{sec:tpm}
}

PyPhi supports three different \tpm representations: \(2\)-\emph{dimensional state-by-node}, \emph{multidimensional state-by-node}, and \emph{state-by-state}. The state-by-node form is the canonical representation in PyPhi, with the 2-dimensional form used for input and visualization and the multidimensional form used for internal computation. The state-by-state representation is given as an input option for those accustomed to this more general form. If the \tpm is given in state-by-state form, PyPhi will raise an error if it does not satisfy \cref{eq:conditional-independence} (conditional independence).

\hypertarget{sec:2d-sbn-form}{%
\subsubsection{2-dimensional state-by-node form}\label{sec:2d-sbn-form}
}

A \tpm in state-by-node form is a matrix where the entry \((i,\,j)\) gives the probability that the \(j^{\mathrm{th}}\) node will be ON at \(t+1\) if the system is in the \(i^{\mathrm{th}}\) state at \(t\). This representation has the advantage of being more compact than the state-by-state form, with \(2^n \by n\) entries instead of \(2^n \by 2^n\), where \(n\) is the number of nodes. Note that the \tpm admits this representation because in PyPhi the nodes are binary; both \(\Pr(N_{t+1} = \mathrm{ON})\) and \(\Pr(N_{t+1} = \mathrm{OFF})\) can be specified by a single entry, in our case \(\Pr(N_{t+1} = \mathrm{ON})\), since the two probabilities must sum to \(1\).

Because the possible system states at \(t\) are represented implicitly as row indices in 2-dimensional {\tpm}s, there must be an implicit mapping from states to indices. In PyPhi this mapping is achieved by listing the state tuples in lexicographical order and then interpreting them as binary numbers where the state of the first node corresponds to the least-significant bit, so that \eg the state \texttt{(0, 0, 0, 1)} is mapped to the row with index \(8\) (the ninth row, since Python uses zero-based indexing~\cite{dijkstra1982numbering}). Designating the first node's state as the least-significant bit is analogous to choosing the little-endian convention in organizing computer memory. This convention is preferable because the mapping is stable under the inclusion of new nodes: including another node in a subsystem only requires concatenating new rows and a new column to its \tpm rather than interleaving them. Note that this is opposite convention to that used in writing numbers in positional notation; care must be taken when converting between states and indices and between different \tpm representations (the \texttt{pyphi.convert} module provides convenience functions for these purposes).

\hypertarget{sec:nd-form}{%
\subsubsection{Multidimensional state-by-node form}\label{sec:nd-form}
}

When a state-by-state \tpm is provided to PyPhi by the user, it is converted to state-by-node form and the conditional independence property (\cref{eq:conditional-independence}) is checked. Note that any \tpm in state-by-node form necessarily satisfies \cref{eq:conditional-independence}. For internal computations, the \tpm is then reshaped so that it has \(n+1\) dimensions rather than two: the first \(n\) dimensions correspond to the states of each of the \(n\) nodes at \(t\), while the last dimension corresponds to the probabilities of each node being ON at \(t+1\). In other words, the indices of the rows (current states) in the 2-dimensional \tpm are ``unraveled'' into \(n\) dimensions, with the \(i^\mathrm{th}\) dimension indexed by the \(i^\mathrm{th}\) bit of the 2-dimensional row index according to the little-endian convention. Because the \tpm is stored in a NumPy array, this multidimensional form allows us to take advantage of NumPy indexing~\cite{walt2011numpy} and use a state tuple as an index directly, without converting it to an integer index:

\begin{samepage}
\begin{Highlighting}
  >>> state = (0, 1, 1)
  >>> print(tpm[state])
  [ 1.    0.25  0.75]
\end{Highlighting}
\end{samepage}

The first entry of this array signifies that if the state of the system is \texttt{(0, 1, 1)} at \(t\), then the probability of the first node \(N_0\) being ON at \(t+1\) is \(\Pr(N_{0,\,t+1} = \mathrm{ON}) = 1\). Similarly, the second entry means \(\Pr(N_{1,\,t+1} = \mathrm{ON}) = 0.25\) and the third entry means \(\Pr(N_{2,\,t+1} = \mathrm{ON}) = 0.75\). 

Most importantly, the multidimensional representation simplifies the calculation of marginal and conditional distributions and cause/effect repertoires, because it allows efficient ``broadcasting''~\cite{walt2011numpy} of probabilities when multiplying distributions. Specifically, the Python multiplication operator `\texttt{*}' acts as the tensor product when the operands are NumPy arrays \texttt{A} and \texttt{B} of equal dimensionality such that for each dimension \texttt{d}, either \texttt{A.shape[d] == 1} or \texttt{B.shape[d] == 1}.

\hypertarget{sec:repertoires}{%
\subsection{Calculation of cause/effect repertoires from the \tpm}\label{sec:repertoires}
}

The cause and effect repertoires of a mechanism over a purview describe how the mechanism nodes in a particular state at \(t\) constrain the possible states of the purview nodes at \(t-1\) and \(t+1\), respectively. Here we describe how they are derived from the \tpm in PyPhi.

\subsubsection{The effect repertoire}

We begin with the simplest case: calculating the effect repertoire of a mechanism \(M \subseteq S\) over a purview consisting of a single element \(P_i \in S\). This is defined as a conditional probability distribution over states of the purview element at \(t+1\) given the current state of the mechanism,
\begin{equation}
  \effectrepertoire(M,\,P_i) \;\coloneqq\; \Pr(P_{i,t+1} \mid M_t = m_t).
\label{eq:effect-repertoire}
\end{equation}
It is derived from the \tpm by conditioning on the state of the mechanism elements, marginalizing over the states of non-purview elements \(P' = S \setminus P_i\) (these states correspond to columns in the state-by-state \tpm), and marginalizing over the states of non-mechanism elements \(M' = S \setminus M\) (these correspond to rows):
\[
\begin{multlined}
  \Pr(P_{i,t+1} \mid M_t = m_t) \;= \\[6pt]
    \dfrac{1}{\lvert \Omega_{M'} \rvert} \;
    \sum_{m'_t\,\in\,\Omega_{M'}} \; 
    \dfrac{1}{\lvert \Omega_{P'} \rvert}\; 
    \sum_{p'_{t+1}\,\in\,\Omega_{P'}} \;
    \Pr(P_{i,t+1},\, p'_{t+1} \mid M = m_t,\, M' = m'_t).
\end{multlined}
\]

This operation is implemented in PyPhi by several subroutines. First, in a pre-processing step performed when the \texttt{Subsystem} object is created, a \texttt{Node} object is created for each element in the subsystem. Each \texttt{Node} contains its own individual \tpm, extracted from the subsystem's \tpm; this is a \(2^s \by 2\) matrix where \(s\) is the number of the node's parents and the entry \((i, j)\) gives the probability that the node is in state \(j\) (\texttt{0} or \texttt{1}) at \(t + 1\) given that its parents are in state \(i\) at \(t\). This node TPM is represented internally in multidimensional state-by-node form as usual, with singleton dimensions for those subsystem elements that are not parents of the node. The effect repertoire is then calculated by conditioning the purview node's \tpm on the state of the mechanism nodes that are also parents of the purview node, via the \texttt{pyphi.tpm.condition\_tpm()} function, and marginalizing out non-mechanism nodes, with \texttt{pyphi.tpm.marginalize\_out()}.

In cases where there are mechanism nodes that are not parents of the purview node, the resulting array is multiplied by an array of ones that has the desired shape (dimensions of size two for each mechanism node, and singleton dimensions for each non-mechanism node). Because of NumPy's broadcasting feature, this step is equivalent to taking the tensor product of the array with the maximum-entropy distribution over mechanism nodes that are not parents, so that the final result is a distribution over all mechanism nodes, as desired.

The effect repertoire over a purview of more than one element is given by the tensor product of the effect repertories over each individual purview element,
\begin{equation}
  \effectrepertoire(M,\,P) \;\coloneqq\; 
    \bigotimes_{P_i \,\in\, P} \effectrepertoire(M,\,P_i).
\label{eq:multinode-effect-repertoire}
\end{equation}
Again, because PyPhi {\tpm}s and repertoires are represented as tensors (multidimensional arrays), with each dimension corresponding to a node, the NumPy multiplication operator between distributions over different nodes is equivalent to the tensor product. Thus the effect repertoire over an arbitrary purview is trivially implemented by taking the product of the effect repertoires over each purview node with \texttt{numpy.multiply()}.

\subsubsection{The cause repertoire}

The cause repertoire of a single-element mechanism \(M_i \in S\) over a purview \(P \subseteq S\) is defined as a conditional probability distribution over the states of the purview at \(t-1\) given the current state of the mechanism,
\begin{equation}
  \causerepertoire(M_i,\,P) \;\coloneqq\; \Pr(P_{t-1} \mid M_{i,t} = m_{i,t}).
\label{eq:cause-repertoire}
\end{equation}
As with the effect repertoire, it is obtained by conditioning and marginalizing the \tpm. However, because the \tpm gives conditional probabilities of states at \(t+1\) given the state at \(t\), Bayes' rule is first applied to express the cause repertoire in terms of a conditional distribution over states at \(t - 1\) given the state at \(t\), 
\[
  \Pr(P_{t-1} \mid M_{i,t} = m_{i,t})
  \;=\; 
  \dfrac{ \Pr(m_{i,t} \mid P_{t-1}) \; \Pr(P_{t-1}) }{ \Pr(m_{i,t}) }.
\]
where the marginal distribution \(\Pr(P_{t-1})\) over previous states is the uniform distribution. In this way, the analysis captures how a mechanism in a state constrains a purview without being biased by whether certain states arise more frequently than others in the dynamical evolution of the system~\cite{ay2008information,balduzzi2008integrated,hoel2013quantifying, oizumi2014phenomenology}. Then the cause repertoire can be calculated by marginalizing over the states of non-mechanism elements \(M' = S \setminus M_i \) (now corresponding to columns in the state-by-state \tpm) and non-purview elements \(P' = S \setminus P\) (now corresponding to rows),
\[
\begin{multlined}
  \dfrac{\Pr(m_{i,t} \mid P_{t-1}) \; \Pr(P_{t-1}) }{ \Pr(m_{i,t}) } \\[10pt]
  \begin{aligned}
    &\;=\; \dfrac{\displaystyle \left( \frac{1}{\lvert \Omega_{P'} \rvert}\; \sum_{p'_{t-1}\,\in\,\Omega_{P'}}\; \frac{1}{\lvert \Omega_{M'} \rvert}\; \sum_{m'_t\,\in\,\Omega_{M'}} \; \Pr(m_{i,t},\, m'_t \mid P_{t-1},\, P'_{t-1} = p'_{t-1}) \right) \Pr(P_{t-1}) }{ \displaystyle \frac{1}{\lvert \Omega_{M'} \rvert}\; \sum_{m'_t\,\in\,\Omega_{M'}}\; \Pr(m_{i,t},\, m'_t) } \\
    &\;=\; \dfrac{\displaystyle \left( \frac{1}{\lvert \Omega_{P'} \rvert}\; \sum_{p'_{t-1}\,\in\,\Omega_{P'}}\; \sum_{m'_t\,\in\,\Omega_{M'}} \; \Pr(m_t,\, m'_t \mid P_{t-1},\, P'_{t-1} = p'_{t-1}) \right) \Pr(P_{t-1}) }{ \displaystyle \sum_{m'_t\,\in\,\Omega_{M'}}\; \Pr(m_{i,t},\, m'_t) }.
  \end{aligned}
\end{multlined}
\] 

In PyPhi, the “backward” conditional probabilities \(\Pr(m_{i,t} \mid P_{t-1})\) for a single mechanism node are obtained by indexing into the last dimension of the node's \tpm with the state $m_{i,t}$ and then marginalizing out non-purview nodes via \texttt{pyphi.tpm.marginalize\_out()}. As with the effect repertoire, the resulting array is then multiplied by an array of ones with the desired shape in order to obtain a distribution over the entire purview. Finally, because in this case the probabilities were obtained from columns of the \tpm, which do not necessarily sum to 1, the distribution is normalized with \texttt{pyphi.distribution.normalize()}.

The cause repertoire of a mechanism with multiple elements is the normalized tensor product of the cause repertoires of each individual mechanism element, 
\begin{equation}
  \causerepertoire(M,\,P) \;=\; \frac{1}{K} \; \bigotimes_{M_i \,\in\, M} \causerepertoire(M_i,\,P),
\label{eq:multinode-cause-repertoire}
\end{equation} 
where 
\[
  K \;=\; \sum_{p_{t-1}\,\in\,\Omega_P} \; \prod_{m_{i,t} \,\in\, \Omega_M} \Pr(P_{t-1} = p_{t-1} \mid M_{i,t} = m_{i,t})
\] 
is a normalization factor that ensures that the distribution sums to \(1\). This is implemented in PyPhi via \texttt{numpy.multiply()} and \texttt{pyphi.distribution.normalize()}. For a more complete illustration of these procedures, see \siref{sup:presentation}.

\hypertarget{sec:organization}{%
\subsection{Code organization and interface design}\label{sec:organization}
}

The postulates of \iit induce a natural hierarchy of computations~\cite[Supplementary Information S2]{tononi2016integrated}, and PyPhi's implementation mirrors this hierarchy by using object-oriented programming (\cref{tab:correspondence}) and factoring the computations into compositions of separate functions where possible. One advantage of this approach is that each level of the computation can be performed independently of the higher levels; for example, if one were interested only in the \mie of certain mechanisms rather than the full \mics, then one could simply call \texttt{Subsystem.effect\_mip()} on those mechanisms instead of calling \computesia and extracting them from the resulting \pyphisia object (this is especially important in the case of large systems where the full calculation is infeasible). Separating the calculation into many subroutines and exposing them to the user also has the advantage that they can be easily composed to implement functionality that is not already built-in.

\begin{table}[!ht]
\centering
\caption{
{\bf Correspondence between theoretical objects and PyPhi objects.}}\label{tab:correspondence}
{\renewcommand{\arraystretch}{1.5}
\begin{tabular}{r|>{\raggedright\arraybackslash}p{3in}}
\textbf{Theoretical object}      & \textbf{PyPhi object}                                                               \\ \thickhline
Discrete dynamical system        & \texttt{Network}                                                                    \\ \hline
Candidate system                 & \texttt{Subsystem}                                                                  \\ \hline
System element                   & \texttt{Node} in \texttt{Subsystem.nodes}                                           \\ \hline
System state                     & Python \texttt{tuple} containing a \texttt{0} or \texttt{1} for each node           \\ \hline
Mechanism                        & Python \texttt{tuple} of node indices                                               \\ \hline
Purview                          & Python \texttt{tuple} of node indices                                               \\ \hline
Repertoire over a purview \(P\)  & NumPy array with \(\lvert P \rvert\) dimensions, each of size 2                     \\ \hline
\mip                             & The \texttt{partition} attribute of the \texttt{RepertoireIrreducibilityAnalysis} 
                                   returned by \texttt{Subsystem.cause\_mip()} or \texttt{Subsystem.effect\_mip()}     \\ \hline
\mic and \mie                    & \texttt{MaximallyIrreducibleCause} and \texttt{MaximallyIrreducibleEffect}          \\ \hline 
Concept                          & \pyphiconcept                                                                       \\ \hline
\(\varphi\)                      & The \texttt{phi} attribute of a \pyphiconcept                                       \\ \hline 
\ces                             & \pyphices                                                                           \\ \hline 
\(\Phi\)                         & The \texttt{phi} attribute of a \pyphices                                           \\ \hline 
\mics                            & The \texttt{ces} attribute of the \pyphisia returned by \computemajorcomplex        \\ \hline 
Complex                          & The \texttt{subsystem} attribute of the \pyphisia returned by \computemajorcomplex
\end{tabular}}
\end{table}

\hypertarget{sec:config}{%
\subsection{Configuration}\label{sec:config}
}

Many aspects of PyPhi's behavior may be configured via the \texttt{pyphi.config} object. The configuration can be specified in a YAML file~\cite{yaml}; an \href{https://github.com/wmayner/pyphi/blob/master/pyphi_config.yml}{example} is available in the GitHub repository. When PyPhi is imported, it checks the current directory for a file named \texttt{pyphi\_config.yml} and automatically loads it if it exists. Configuration settings can also be loaded on the fly from an arbitrary file with the \texttt{pyphi.config.load\_config\_file()} function.

Alternatively, \texttt{pyphi.config.load\_config\_dict()} can load configuration settings from a Python dictionary. Many settings can also be changed by directly assigning them a new value.

Default settings are used if no configuration is provided. A full description of the available settings and their default values is available in the \href{https://pyphi.readthedocs.io/page/configuration}{``Configuration'' section of the online documentation}. 

\hypertarget{sec:optimizations}{%
\subsection{Optimizations and approximations}\label{sec:optimizations}
}

Here we describe various optimizations and approximations used by the software to reduce the complexity of the calculations (see \sectionref{sec:limitations}). Memoization and caching optimizations are described in \siref{sup:caching}.

\hypertarget{sec:connectivity}{%
\subsubsection{Connectivity optimizations}\label{sec:connectivity}
}

As mentioned in \sectionref{sec:input}, providing connectivity information explicitly with a \cm can greatly reduce the time complexity of the computations, because in certain cases missing connections imply reducibility \emph{a priori}.

For example, at the system level, if the subsystem is not strongly connected then \(\Phi\) is necessarily zero. This is because a unidirectional cut between one system part and the rest can always be found that will not actually remove any edges, so the {\ces}s with and without the cut will be identical (see \siref{sup:strong-connectivity} for proof). Accordingly, PyPhi immediately excludes these subsystems when finding the major complex of a system.

Similarly, at the mechanism level, PyPhi uses the \cm to exclude certain purviews from consideration when computing a \mic or \mie by efficiently determining that repertoires over those purviews are reducible without needing to explicitly compute them. Suppose there are two sets of nodes \(X\) and \(Y\) for which there exist partitions \(X = (X_1,\; X_2)\) and \(Y = (Y_1,\; Y_2)\) such that there are no edges from \(X_1\) to \(Y_2\) and no edges from \(X_2\) to \(Y_1\). Then the effect repertoire of mechanism \(X\) over purview \(Y\) can be factored as 
\[
\begin{multlined}
  \effectrepertoire(X, Y) \;=\; \\
    \hspace*{3em} \effectrepertoire(X_1, Y_1) \;\otimes\; \effectrepertoire(X_2, Y_2),
\end{multlined}
\] 
and the cause repertoire of mechanism \(Y\) over purview \(X\) can be factored as 
\[
\begin{multlined}
  \causerepertoire(Y, X) \;=\; \\
    \hspace*{3em} \causerepertoire(Y_1, X_1) \;\otimes\; \causerepertoire(Y_2, X_2).
\end{multlined}
\] 
Thus in these cases the mechanism is reducible for that purview and \(\varphi = 0\) (see \siref{sup:block-factorable} for proof).

\hypertarget{sec:emd}{%
\subsubsection{Analytical solution to the earth mover's distance}\label{sec:emd}
}

One of the repertoire distances available in PyPhi is the earth mover's distance (\emd), with the Hamming distance as the ground metric. Computing the \emd between repertoires is a costly operation, with time complexity \(O(n 2^{3n})\) where \(n\) is the number of nodes in the purview~\cite{pele2009fast}. However, when comparing effect repertoires, PyPhi exploits a theorem that states that the \emd between two distributions $p$ and $q$ over multiple nodes is the sum of the {\emd}s between the marginal distributions over each individual node, if $p$ and $q$ are independent. This analytical solution has time complexity \(O(n)\), a significant improvement over the general \emd algorithm (note that this estimate does not include the cost of computing the marginals, which already have been computed to obtain the repertoires). By the conditional independence property (Eq.~\ref{eq:conditional-independence}), the conditions of the theorem hold for \emd calculations between effect repertoires, and thus the analytical solution can be used for half of all repertoire calculations performed in the analysis. The theorem is formally stated and proved in \siref{sup:emd}.

\hypertarget{sec:approximations}{%
\subsubsection{Approximations}\label{sec:approximations}
}

Currently, two approximate methods of computing \(\Phi\) are available. These can be used via settings in the PyPhi configuration file (they are disabled by default):

\begin{enumerate}
\def\labelenumi{\arabic{enumi}.}
\item
  \texttt{pyphi.config.CUT\_ONE\_APPROXIMATION} (the ``cut one'' approximation), and
\item
  \texttt{pyphi.config.ASSUME\_CUTS\_CANNOT\_CREATE\_NEW\_CONCEPTS} (the ``no new concepts'' approximation).
\end{enumerate}

In both cases, the complexity of the calculation is greatly reduced by replacing the optimal partitioned \ces by an approximate solution. The system's \(\Phi\) value is then computed as usual as the difference between the unpartitioned \ces and the approximate partitioned \ces.

\paragraph*{Cut one.} 
The ``cut one'' approximation reduces the scope of the search for the \mip over possible system cuts. Instead of evaluating the partitioned \ces for each of the \(2^n\) unidirectional bipartitions of the system, only those \(2n\) bipartitions are evaluated that sever the edges from a single node to the rest of the network or vice versa. Since the goal is to find the minimal \(\Phi\) value across all possible partitions, the ``cut one'' approximation provides an upper bound on the exact \(\Phi\) value of the system.

\paragraph*{No new concepts.} 
For most choices of mechanism partitioning schemes and distance measures, it is possible that the \ces of the partitioned system contains concepts that are reducible in the unpartitioned system and thus not part of the unpartitioned \ces. For this reason, PyPhi by default computes the partitioned \ces from scratch from the partitioned \tpm. Under the ``no new concepts'' approximation, such new concepts are ignored. Instead of repeating the entire \ces computation for each system partition, which requires reevaluating all possible candidate mechanisms for irreducibility, only those mechanisms are taken into account that already specify concepts in the unpartitioned \ces. In many types of systems, new concepts due to the partition are rare. Approximations using the ``no new concepts'' option are thus often accurate. Note, however, that this approximation provides neither a theoretical upper nor lower bound on the exact \(\Phi\) value of the system.

\hypertarget{sec:limitations}{%
\subsection*{Limitations}\label{sec:limitations}
}

PyPhi's main limitation is that the algorithm is exponential time in the number of nodes, \(O(n 53^n)\). This is because the number of states, subsystems, mechanisms, purviews, and partitions that must be considered each grows exponentially in the size of the system. This limits the size of systems that can be practically analyzed to {\textasciitilde}10--12 nodes. For example, calculating the major complex of systems of three, five, and seven stochastic majority gates, connected in a circular chain of bidirectional edges, takes {\textasciitilde}1 s, {\textasciitilde}16 s, and {\textasciitilde}2.75 h respectively (parallel evaluation of system cuts, \(32\times3.1\)GHz CPU cores). Using the ``cut one'' approximation, these calculations take {\textasciitilde}1 s, {\textasciitilde}12 s, and {\textasciitilde}0.63 h. In practice, actual execution times are substantially variable and depend on the specific network under analysis, because network structure determines the effectiveness of the optimizations discussed above.

Another limitation is that the analysis can only be meaningfully applied to a system that is Markovian and satisfies the conditional independence property. These are reasonable assumptions for the intended use case of the software: analyzing a causal \tpm derived using the calculus of perturbations~\cite{pearl2009causality}. However, there is no guarantee that these assumptions will be valid in other circumstances, such as {\tpm}s derived from observed time series (\eg, EEG recordings). Whether a system has the Markov property and conditional independence property should be carefully checked before applying the software in novel contexts.

\hypertarget{sec:availability}{%
\section{Availability and future directions}\label{sec:availability}
}

PyPhi can be installed with Python's package manager via the command `\texttt{pip~install~pyphi}' on Linux and macOS systems equipped with Python 3.4 or higher. It is open-source and licensed under the GNU General Public License v3.0. The source code is version-controlled with \texttt{git} and hosted in a public repository on GitHub at \texttt{\small\url{https://github.com/wmayner/pyphi}}. Comprehensive and continually-updated documentation is available online at \texttt{\small\url{https://pyphi.readthedocs.io}}. The \texttt{pyphi-users} mailing list can be joined at \texttt{\small\url{https://groups.google.com/forum/\#!forum/pyphi-users}}. A web-based graphical interface to the software is available at \texttt{\small\url{http://integratedinformationtheory.org/calculate.html}}.

Several additional features are in development and will be released in future versions. These include a module for calculating \(\Phi\) over multiple spatial and temporal scales, as theoretically required by the exclusion postulate (in the current version, the \texttt{Network} is assumed to represent the system at the spatiotemporal timescale at which \(\Phi\) is maximized~\cite{hoel2016can,marshall2016black}), and a module implementing a calculus for ``actual causation'' as formulated in~\cite{albantakis2017actual} (preliminary versions of these modules are available in the current release). The software will also be updated to reflect developments in \iit and further optimizations in the algorithm.

\section{Acknowledgments}

We thank Andrew Haun, Leonardo Barbosa, Sabrina Streipert, and Erik Hoel for their valuable feedback as early users of the software.

WGPM, WM, LA, RM and GT received funding from the Templeton World Charities Foundbtion (Grant \#TWCF 0067/AB41). WGPM and GF were also supported by National Research Service Award (NRSA) T32 GM007507.

\section{Supporting information}
\vspace{6pt}

\paragraph*{S1~Calculating~\(\Phi\).}\label{sup:presentation}
{\bf Illustration of the algorithm.}

\paragraph*{S2~Appendix.}\label{sup:caching}
{\bf Memoization and caching optimizations.}

\paragraph*{S3~Appendix.}\label{sup:strong-connectivity}
{\bf Proof of the strong connectivity optimization.}

\paragraph*{S4~Appendix.}\label{sup:block-factorable}
{\bf Proof of the block-factorable optimization.}

\paragraph*{S5~Appendix.}\label{sup:emd}
{\bf Proof of an analytical solution to the \emd between effect repertoires.}

\label{mylastpage}
\newpage
\setcounter{page}{1}
\rfoot{\thepage/\pageref{sup2last}}
\subsection*{S2 Appendix}
\subsubsection*{Memoization and caching optimizations.}

During the course of computing a \pyphisia, several functions in PyPhi are called multiple times with the same input. For example, calculating \texttt{cause\_repertoire((A,B),\ (A,B,C)} and \texttt{cause\_repertoire((A,C),\ (A,B,C))} both require calculating \texttt{cause\_repertoire((A,),\ (A,B,C))}. Similarly, \texttt{cause\_repertoire((A,),\ (B,C))} is both the unpartitioned repertoire of the candidate MIP of mechanism \(A\) over purview \(BC\) and the first term in the expression for the partitioned repertoire of the candidate MIP of mechanism \(AB\) over purview \(ABC\) under the partition 
\[ 
  \frac{A}{BC} \times \frac{B}{A}.
\]
In such situations, a natural optimization technique to reduce expensive re-computation of these functions is \emph{memoization}: when a function is computed for a given input, the input-output pair is stored in a lookup table; if the function is called again with that input, the output is simply looked up in the table and returned, without computing the function again.

In PyPhi, memoization is applied to various functions at different levels of the algorithm, listed here:
\begin{itemize}
  \item \texttt{Subsystem.\_single\_node\_cause\_repertoire()} and \texttt{Subsystem.\_single\_node\_effect\_repertoire()}, the un-normalized cause repertoire of a single-node mechanism and the effect repertoire over a single-node purview, respectively (note that these functions are meant to be called internally by other PyPhi functions and not by the user, as indicated by the leading underscore);
  \item \texttt{Subsystem.cause\_repertoire()} and \texttt{Subsystem.effect\_repertoire()}, the cause and effect repertoires of arbitrary mechanism-purview pairs;
  \item \texttt{Subsystem.mic()} and \texttt{Subsystem.mie()}, the \mic and \mie of a mechanism;
  \item \texttt{Network.potential\_purviews()}, the purviews which are not necessarily reducible based on the \cm;
  \item Various utility functions such as \texttt{pyphi.distribution.max\_entropy\_distribution()}; and
  \item \computesia, the full \iit analysis of a \texttt{Subsystem}.
\end{itemize}

At the highest level, \computesia is memoized such that the \pyphisia object is stored persistently on the filesystem, rather than in memory, in a directory named \texttt{\_\_pyphi\_cache\_\_} (which is automatically created in the directory where the Python session was started). This means that the \pyphisia objects are automatically saved across Python sessions and can be quickly retrieved simply by running the same code, which is useful for interactive, exploratory work. This behavior can be controlled with the \texttt{pyphi.config.CACHE\_SIAS} configuration setting. Note, however, that this feature is primarily for convenience and is not intended to replace explicit data management. Additionally, care must be taken to erase or disable the cache when upgrading to new versions of PyPhi, as changes to the algorithm may invalidate previously computed output. A final caveat: because the results are stored on the filesystem, they can accumulate and occupy a large amount of disk space if the \texttt{\_\_pyphi\_cache\_\_} directory is not periodically removed.

\label{sup2last}
\newpage
\setcounter{page}{1}
\rfoot{\thepage/\pageref{sup3last}}
\subsection*{S3 Appendix}
\subsubsection*{Proof of the strong connectivity optimization.}

\begin{theorem*}[strong connectivity]
If \(G = (V, E)\) is a directed graph that is not strongly connected, then \(\Phi(G) = 0\).
\end{theorem*}

\begin{proof}

Since \(G\) is not strongly connected, \(G\) contains \(n \geq 2\) strongly connected components, which we arbitrarily label 
\[
  G_1, G_2, \ldots, G_n \;=\; (V_1, E_1), (V_2, E_2), \ldots, (V_n, E_n).
\]
Let \(E_{j,k}\) denote the set of directed edges from nodes in component \(G_j\) to those in component \(G_k\), \(\{(a, b) \in E \mid a \in V_j \text{ and } b \in V_k\}\).

Consider the first component \(G_1\). For every other component \(G_i,\; i \neq 1\), either \(E_{1,i} = \varnothing\) or \(E_{i,1} = \varnothing\), because otherwise \(G_1\) and \(G_i\) would not be distinct connected components. Now let \(\overline{G_1}\) be the indices of components that receive no edges from \(G_1\), \(\{i \in {1,\ldots,n} \mid E_{1,i} = \varnothing \}\). Then let \(Y\) be the union of the nodes in these components, \[ Y = \bigcup_{i\;\in\;\overline{G_1}} V_i\;, \] and let \(X = V \setminus Y\). Then \(X\) and \(Y\) form a partition of \(V\) such that there are no edges from any nodes in \(X\) to any nodes in \(Y\).

Now consider the system cut \(c(X, Y)\) that cuts edges from nodes in \(X\) to nodes in \(Y\). Because there are no such edges, none of the node {\tpm}s are changed after applying the cut, and thus the subsystem \tpm is unchanged because it is the product of the node {\tpm}s. Since the cause-effect structure of a system is a function of the subsystem's \tpm, the cause-effect structure \(C_{c(X,Y)}\) of the partitioned subsystem and the cause-effect structure \(C\) of the unpartitioned subsystem are identical.

Let \(\Phi_{c(X,Y)}(G)\) be the \(\Phi\) value of \(G\) with respect to \(c(X,Y)\). By definition, this is the \texttt{ces\_distance} between \(C\) and \(C_{c(X,Y)}\). The \texttt{ces\_distance} function is a metric, so since \(C_{c(X,Y)} = C\) we have that 
\[
\texttt{ces\_distance}(C, C_{c(X,Y)}) \;=\; 0
\] 
by non-negativity of metrics, and thus \(\Phi_{c(X,Y)}(G) = 0\). Now, by definition, \[ \Phi(G) \;=\; \min_{c \;\in\; \mathbb{C}} \; \Phi_c(G) \] where \(\mathbb{C}\) is the set of all system cuts. Since \(\Phi_{c(X,Y)}(G) = 0\), by non-negativity we have \(\Phi(G) = 0\).

\end{proof}

\label{sup3last}
\newpage
\setcounter{page}{1}
\rfoot{\thepage/\pageref{sup4last}}
\subsection*{S4 Appendix}
\subsubsection*{Proof of the block-factorable optimization.}

\begin{definition*}[block diagonal]
An \(n \by m\) matrix \(\mathbf{A}_{n,m}\) is said to be \emph{block diagonal} if it can be written as 
\[ 
  \mathbf{A}_{n,m} \;=\;
    \begin{bmatrix*}[l]
        \mathbf{B}_{s,t}      & \mathbf{0}_{s, m - t} \\ 
        \mathbf{0}_{n - s, t} & \mathbf{C}_{n - s, m - t} 
    \end{bmatrix*}, 
\] 
where \(1 \leq s < n\) and \(1 \leq t < m\).
\end{definition*}

We consider sub-matrices of the connectivity matrix \(\cmmath\) of the form \(\cmmath(\pi_{\textrm{row}}, \pi_{\textrm{col}})\), where 
\[
  {\left[\cmmath(\pi_{\textrm{row}}, \pi_{\textrm{col}})\right]}_{i, j} \;=\; {\left[\cmmath\right]}_{\pi_{\textrm{row}}(i), \pi_{\textrm{col}}(j)}. 
\]

\begin{definition*}[block factorable]
A mechanism-purview pair \((M, P)\) is said to be \emph{block factorable} if there exists a permutation \(\pi_M\) of the mechanism indices and a permutation \(\pi_P\) of the purview indices such that \(\cmmath(\pi_M, \pi_P)\) is block diagonal (for effect purviews) or \(\cmmath(\pi_P, \pi_M)\) is block diagonal (for cause purviews).
\end{definition*}

\begin{theorem*}[block reducibility]
If a mechanism-purview pair is block factorable, then it is reducible \((\varphi = 0)\).
\end{theorem*}

\begin{proof}

Consider a mechanism \(M\) constituted of \(n\) elements and a purview \(P\) constituted of \(m\) elements. Assume without loss of generality that \(P\) is an effect purview. Since \((M, P)\) is block factorable, there exist permutations \(\pi_M\) and \(\pi_P\) such that \(\cmmath(\pi_M, \pi_P)\) is block diagonal, \emph{i.e.}, 
\[
  \cmmath(\pi_M, \pi_P) \;=\;
    \begin{bmatrix*}[l]
      \mathbf{B}_{s,t}      & \mathbf{0}_{s, m - t} \\ 
      \mathbf{0}_{n - s, t} & \mathbf{C}_{n - s, m - t} 
    \end{bmatrix*}, 
\] 
where \(1 \leq s < n\) and \(1 \leq t < m\).

We define a mechanism-purview partition 
\[
  c \;\coloneqq\; \frac{M_1}{P_1} \times \frac{M_2}{P_2} \,
\] 
that cuts edges from \(M_1\) to \(P_2\) and from \(M_2\) to \(P_1\), where 
\begin{alignat*}{10}
  M_1 \;&=\; \{\; & \pi_M(i) &\mid&\; 1     \,&\leq&\; i &<&\; s \,& + \,1 & \;\} \\
  M_2 \;&=\; \{\; & \pi_M(i) &\mid&\; s + 1 \,&\leq&\; i &<&\; n \,& + \,1 & \;\} \\
  P_1 \;&=\; \{\; & \pi_P(i) &\mid&\; 1     \,&\leq&\; i &<&\; t \,& + \,1 & \;\} \\
  P_2 \;&=\; \{\; & \pi_P(i) &\mid&\; t + 1 \,&\leq&\; i &<&\; m \,& + \,1 & \;\}
\end{alignat*} 
Note that \([\cmmath(\pi_M, \pi_P)]_{i,j} = 0\) if either \(i \in M_1\) and \(j \in P_2\) or \(i \in M_2\) and \(j \in P_1\). Thus there are no edges cut by \(c\), and it leaves the subsystem's \tpm unchanged. Since the effect repertoire of a mechanism-purview combination is a function of the subsystem's \tpm, the unpartitioned effect repertoire, \(\er(M, P)\) and the partitioned repertoire \(\er_c(M, P)\) are identical.

By definition, \(\varphi_c(M, P)\) is the distance between \(\er(M, P)\) and \(\er_c(M, P)\), so \(\varphi_c(M, P) = 0\). Now, by definition, \[ \varphi(M, P) =\; \min_{p \;\in\; \mathbb{P}} \; \varphi_p(M, P), \] where \(\mathbb{P}\) is the set of all partitions of \((M, P)\). Since \(\varphi_c(M, P) = 0\), by the non-negativity of metrics we have \(\varphi(M, P) = 0\). 

\end{proof}

\label{sup4last}
\newpage
\setcounter{page}{1}
\rfoot{\thepage/\pageref{sup5last}}
\subsection*{S5 Appendix}
\subsubsection*{Proof of an analytical solution to the \emd between effect repertoires.}

\begin{theorem*}[analytical \emd]
Consider two random variables \(X_1, X_2\) with corresponding state
spaces \(\Omega_1, \Omega_2\) and an `additive' metric \(D\), 
\[
  D((i_1, i_2),\, (j_1, j_2)) \;=\;
    D(i_1, j_1) + D(i_2, j_2) \quad 
    \forall\ (i_1, i_2)\ \text{and}\ (j_1, j_2) \in \Omega_1 \times \Omega_2.
\] 
Let \(p_1\) and \(q_1\) be two probability distributions on \(X_1\), and let \(p_2\) and \(q_2\) be probability distributions on \(X_2\). If \(X_1\) and \(X_2\) are independent, then the \emd between the joint distributions \(p = p_1p_2\) and \(q = q_1q_2\), with \(D\) as the ground metric, is equal to the sum of the {\emd}s between the marginal distributions: 
\[
  \emdmath(p, q) \;=\; \emdmath(p_1, q_1) + \emdmath(p_2, q_2).
\]
\end{theorem*}

\begin{proof}

First, we demonstrate that 
\[
  \emdmath(p, q) \;\leq\; \emdmath(p_1, q_1) + \emdmath(p_2, q_2).
\]
To do this, we define a third probability distribution as an intermediate point, 
\[
  r \;\coloneqq\; q_1p_2.
\]
We define the following flow from \(p\) to \(r\), 
\[ 
  f_{p,r}(i_1, i_2, j_1, j_2) \;\coloneqq\;
    \begin{cases}
      p_2(i_2)f_{p_1,q_1}^*(i_1, j_1) &\text{if}~i_2 = j_2 \\
      0                              &\otherwise,
    \end{cases}
\] 
where \(f^*_{p_1,q_1}\) is the optimal flow for the \emd between \(p_1\) and \(q_1\). With this flow, we have 
\[
\begin{aligned}
  \emdmath(p, r) \;
  &\leq\; \sum_{i_1, i_2, j_1, j_2} f_{p,r}(i_1, i_2, j_1, j_2)D((i_1, i_2),\, (j_1, j_2)) \\
  &=\; \sum_{i_1, i_2, j_1} p_2(i_2)f^*_{p_1,q_1}(i_1, j_1)D(i_1, j_1) \\
  &=\; \sum_{i_2} p_2(i_2) \sum_{i_1, j_1} f^*_{p_1,q_1}(i_1, j_1)D(i_1, j_1) \\
  &=\; \emdmath(p_1, q_1)
\end{aligned}
\] 
We next define a flow from \(r\) to \(q\), 
\[
  f_{r,q}(i_1, i_2, j_1, j_2) \;\coloneqq\;
    \begin{cases}
      q_1(i_1)f_{p_2,q_2}^*(i_2, j_2) &\text{if}~i_1 = j_1 \\ 
      0                              &\otherwise,
    \end{cases}
\] 
where \(f^*_{p_2,q_2}\) is the optimal flow for the \emd between \(p_2\) and \(q_2\). With this flow, we have 
\[
\begin{aligned}
  \emdmath(r, q)  \;
  &\leq\; \sum_{i_1, i_2, j_1, j_2} f_{r,q}(i_1, i_2, j_1, j_2)D((i_1, i_2),\, (j_1, j_2)) \\
  &=\; \sum_{i_1, i_2, j_2} q_1(i_1)f^*_{p_2,q_2}(i_2, j_2)D(i_2, j_2) \\
  &=\; \sum_{i_1} q_1(i_1) \sum_{i_2, j_2} f^*_{p_2,q_2}(i_2, j_2)D(i_2, j_2) \\
  &=\; \emdmath(p_2, q_2)
\end{aligned}
\] 
Then using the triangle inequality (the \emd is a metric), we have 
\[
  \emdmath(p, q) \;\leq\; \emdmath(p, r) + \emdmath(r, q) \;\leq\; \emdmath(p_1, q_1) + \emdmath(p_2, q_2).
\] 
To complete the proof, we next demonstrate that 
\[
  \emdmath(p, q) \;\geq\; \emdmath(p_1, q_1) + \emdmath(p_2, q_2).
\] 
If \(f_{p,q}^*\) is the optimal flow for \(\emdmath(p, q)\), then define a flow between \(p_1\) and \(q_1\), 
\[
  f_1(i_1,j_1) \;\coloneqq\; \sum_{i_2, j_2} f_{p,q}^*(i_1, i_2, j_1, j_2),
\] 
and a flow between \(p_2\) and \(q_2\) 
\[
  f_2(i_2, j_2) \;\coloneqq\; \sum_{i_1, j_1} f_{p,q}^*(i_1, i_2, j_1, j_2).
\] 
Then using the additive property of the ground metric \(D\), 
\[
\begin{aligned}
  \emdmath(p, q) 
  \;&=\; \sum_{i_1, i_2, j_1, j_2} f_{p,q}^*(i_1, i_2, j_1, j_2)D((i_1, i_2),\, (j_1, j_2)) \\
  \;&=\; \sum_{i_1, i_2, j_1, j_2} f_{p,q}^*(i_1, i_2, j_1, j_2)D(i_1, j_1) + \sum_{i_1, i_2, j_1, j_2} f_{p,q}^*(i_1, i_2, j_1, j_2)D(i_2, j_2) \\
    & \begin{multlined}
        =\; \sum_{i_1, j_1} \left(\sum_{i_2, j_2} f_{p,q}^*(i_1, i_2, j_1, j_2) \right)D(i_1, j_1) \\
          \hspace*{3em} + \sum_{i_2, j_2} \left(\sum_{i_1, j_1} f_{p,q}^*(i_1, i_2, j_1, j_2) \right)D(i_2, j_2)
      \end{multlined} \\
  \;&=\; \sum_{i_1, j_1} f_1(i_1,j_1)D(i_1, j_1) + \sum_{i_2, j_2} f_2(i_2,j_2)D(i_2, j_2) \\
  \;&\geq\; \emdmath(p_1, q_1) + \emdmath(p_2, q_2).
\end{aligned}
\] 
Therefore \(\emdmath(p, q) = \emdmath(p_1, q_1) + \emdmath(p_2, q_2)\).

\end{proof}

Perhaps it is worth demonstrating that the flows \(f_1\), \(f_2\), \(f_{p,r}\) and \(f_{r,q}\) satisfy the \emd requirements. Consider \(f_{p,r}\), 
\[
  f_{p,r}(i_1, i_2, j_1, j_2) \;=\;
    \begin{cases} 
      p_2(i_2)f_{p_1,q_1}^*(i_1, j_1) &\text{if}~i_2 = j_2 \\
      0                               &\otherwise, 
    \end{cases}
\] 
where \(f^*_{p_1,q_1}\) is the optimal flow for the \emd between \(p_1\) and \(q_1\).

Since \(p_2(i_2) \geq 0\) (probability) and \(f_{p_1,q_1}^*(i_1, j_1) \geq 0\) (definition of optimal flow), 
\[
  f_{p,r}(i_1, i_2, j_1, j_2) \geq 0. 
\] 
Next, 
\[
\begin{aligned}
  \sum_{j_1, j_2} f_{p,r}(i_1, i_2, j_1, j_2) 
  \;&=\; \sum_{j_1} p_2(i_2)f_{p_1,q_1}^*(i_1, j_1) \\
  \;&=\; q_1(i_1)p_2(i_2) \\
  \;&=\; r(i_1, i_2),
\end{aligned}
\]
and 
\[
\begin{aligned}
  \sum_{i_1, i_2} f_{p,r}(i_1, i_2, j_1, j_2) 
  \;&=\; \sum_{i_1} p_2(j_2)f_{p_1,q_1}^*(i_1, j_1) \\
  \;&=\; q_1(j_1)p_2(i_2) \\
  \;&=\; r(j_1, j_2).
\end{aligned}
\]
Finally, 
\[
\begin{aligned}
  \sum_{i_1, i_2, j_1, j_2} f_{p,r}(i_1, i_2, j_1, j_2) 
  \;&=\; \sum_{i_2} p_2(i_2) \sum_{i_1, j_1} f_{p_1,q_1}^*(i_1, j_1) \\
  \;&=\; \sum_{i_1, j_1} f_{p_1,q_1}^*(i_1, j_1) \\
  \;&=\; 1
\end{aligned}
\] 
Thus \(f_{p,r}\) satisfies the criteria for a potential flow. The others follow similarly.

\label{sup5last}
\end{document}